\documentclass[aps,prl,reprint,longbibliography]{revtex4-2}

\usepackage{amsfonts,amscd,mathrsfs,amsmath,amsthm,amssymb}
\usepackage{mathtools}
\usepackage{xparse}
\usepackage{graphicx}
\usepackage{float}
\usepackage{pifont}
\usepackage[dvipsnames]{xcolor}
\usepackage{colortbl}
\usepackage{multirow}
\usepackage{enumerate}   
\usepackage{comment}
\usepackage{booktabs}
\usepackage{cancel}
\usepackage{lieart}

\usepackage{listings}
\usepackage{xcolor}

\definecolor{codegreen}{rgb}{0,0.6,0}
\definecolor{codegray}{rgb}{0.2,0.2,0.2}
\definecolor{codepurple}{rgb}{0.58,0,0.82}
\definecolor{backcolour}{rgb}{0.95,0.95,0.92}

\lstdefinestyle{mystyle}{
    backgroundcolor=\color{backcolour},   
    commentstyle=\color{codegreen},
    keywordstyle=[1]\color{magenta},
    keywordstyle=[2]\color{blue},
    otherkeywords = {PerfectGroup,CharacterTable,Irr,Norm,ScalarProduct,gap>},
    morekeywords = [1]{PerfectGroup,CharacterTable,Irr,Norm,ScalarProduct,Degree},
    morekeywords = [2]{gap>,>},
    numberstyle=\tiny\color{codegray},
    stringstyle=\color{codepurple},
    basicstyle=\ttfamily\footnotesize,
    breakatwhitespace=false,         
    breaklines=true,                 
    captionpos=b,                    
    keepspaces=true,                 
    numbers=left,                    
    numbersep=5pt,   
    frame = single, 
    showspaces=false,                
    showstringspaces=false,
    showtabs=false,  
    morecomment=[l][\color{codegreen}]{\#},
    alsoletter=?>,
    tabsize=2
}
\usepackage{physics}
\usepackage{bm}
\usepackage{bbm}
\usepackage{caption}

\PassOptionsToPackage{hyphens}{url}\usepackage{hyperref}

\hypersetup{
    colorlinks=true,
    linkcolor=magenta!80!black,  
    citecolor=magenta!80!black,
    breaklinks=true,
}
\usepackage[capitalise]{cleveref}

\theoremstyle{theorem}

\theoremstyle{definition}
\newtheorem{theorem}{Theorem}

\newtheorem*{conjecture*}{Conjecture}

\newtheorem{lemma}{Lemma}

\AddToHook{cmd/appendix/before}{%
  \setcounter{equation}{0}%
  \renewcommand{\theequation}{\thesection\arabic{equation}}%
  \setcounter{theorem}{0}%
  \renewcommand{\thetheorem}{\thesection\arabic{theorem}}%
  \setcounter{lemma}{0}%
  \renewcommand{\thelemma}{\thesection\arabic{lemma}}%
}
\newcommand\numberthis{\addtocounter{equation}{1}\tag{\theequation}}

    \DeclareMathAlphabet{\mathsfit}{T1}{\sfdefault}{\mddefault}{\sldefault}
    \SetMathAlphabet{\mathsfit}{bold}{T1}{\sfdefault}{\bfdefault}{\sldefault}

    \renewcommand{\SU}{\mathrm{SU}}

    \renewcommand{\U}{\mathrm{U}}

\usepackage[new]{old-arrows}

    \newcommand{\ico}{2\mathrm{I}}

\newcommand{\bmsf}[1]{\bm{\mathsf{#1}}}

\renewcommand{\G}{\mathcal{G}}
\newcommand{\Glog}{\mathsf{G}}

\newcommand{\f}{\bmsf{f}}
\renewcommand{\F}{\bmsf{F}}

\newcommand{\irrE}{\bmsf{R}}
\newcommand{\irre}{\bmsf{R}^\downarrow}

\newcommand{\irrechar}{R^\downarrow}

\newcommand{\irrGAP}{\bmsf{\chi}}
\newcommand{\irrGAPchar}{\chi}

\newcommand{\logirr}{\bmsf{\lambda}}

\newcommand{\irrIco}{\bmsf{\pi}}
\newcommand{\irrIcochar}{\pi}

\begin{document}
\title{Quantum Codes from Twisted Unitary $t$-groups}
\thanks{These authors contributed equally to this work.}
\author{Eric Kubischta}
\email{erickub@umd.edu} 
\author{Ian Teixeira}
\email{igt@umd.edu}
\affiliation{Joint Center for Quantum Information and Computer Science,
NIST/University of Maryland, College Park, Maryland 20742 USA}
\begin{abstract}
We introduce twisted unitary $t$-groups, a generalization of unitary $t$-groups under a twisting by an irreducible representation. We then apply representation theoretic methods to the Knill-Laflamme error correction conditions to show that twisted unitary $t$-groups automatically correspond to quantum codes with distance $d=t+1$. By construction these codes have many transversal gates, which naturally do not spread errors and thus are useful for fault tolerance. 
\end{abstract}

\maketitle

\section{Introduction}

There is a rich history of connections between classical $ t $-designs (both spherical and orthogonal) and classical information theory \cite{ConwaySloane}.  Similarly, there are many applications of quantum $ t $-designs (both complex projective and unitary) to quantum information theory, including  tomography \cite{tomography1,tomography2}, randomized benchmarking \cite{benchmarking1}, cryptography \cite{crypto1}, and chaos \cite{chaos1}.

However, until now, no connection has been made between quantum $ t $-designs and quantum error correcting codes, despite the fact that there are deep connections between classical $ t $-designs and classical error-correcting codes, such as the theorem of Assmus and Matteson \cite{ConwaySloane,Assmus_Mattson_1969}.

Among quantum $ t $-designs, unitary $t$-designs are especially commonplace in the quantum information literature, and in the special case that a unitary $t$-design $ \G $ forms a finite group it is called a \textit{unitary $t$-group}. Unitary $t$-groups are well studied \cite{unitary_designs_and_codes,GrossCharacters&tGroups,rho_designs,bannai2020}.

 % https://arxiv.org/pdf/2102.12617.pdf has a bunch of citations for applications

In forging a connection between quantum $ t $-designs and quantum error correcting codes we first review unitary $t$-groups using tools from representation theory; similar techniques were introduced in \cite{GrossCharacters&tGroups,rho_designs,bannai2020}. We then define \textit{twisted unitary $t$-groups} and argue that they are a natural generalization of unitary $t$-groups under a ``twisting" by $\logirr$, an irreducible representation (irrep) of $\G$. In the special case that $\logirr$ is the trivial irrep $\irrep{1}$, twisted unitary $t$-groups are equivalent to the regular notion of unitary $t$-groups. 

\emph{Application. --} Our main application of twisted unitary $t$-groups is in constructing quantum codes that naturally have many transversal logical gates. A logical gate for an $ n $ qudit quantum code is called transversal if it can be implemented as $U_1 \otimes \cdots \otimes U_n$ where each unitary $ U_i $ acts on a single physical qudit. As an application of twisted unitary $t$-groups we show that they induce quantum codes with distance $d = t+1$, and that the corresponding quantum codes can have very large groups of transversal gates: for each $ g \in \G $ the physical gate $ g^{\otimes n } $ implements the logical gate $ \logirr(g) $. A logical gate implemented by the physical gate $ g^{\otimes n} $ acts on each physical qudit separately and so does not spread errors between physical qudits, for this reason such gates are likely to be useful for fault tolerance.  

We highlight our application via two examples. In our first example we show that the unitary $ 5 $-group $\ico$, the binary icosahedral subgroup of $\SU(2)$, forms a twisted unitary $2$-group with respect to a particular 2-dimensional irrep. This yields $n$-qubit codes for all odd $n \geq 7$ of distance $d = 3$ (meaning these codes can correct an arbitrary single error). Each of these codes implements all of $\ico$ transversally. In our second example, we show that the unitary $3$-group $\Sigma(360\phi)$, a maximal subgroup of $\SU(3)$ (and well studied in the high energy literature \cite{SU3Fairbairn1964FiniteAD,SU3Ludl_2011,SUtree}), forms a twisted unitary $1$-group with respect to two distinct $3$-dimensional irreps. These yield $n$-qutrit quantum codes of distance $d=2$ for all $n \geq 5$ not divisible by $ 3 $. Each of these codes implements all of $\Sigma(360\phi)$ transversally. For both our examples an encoding circuit can be obtained using \cite{Lev}.

\emph{Motivation.--} Since transversal gates do not propagate errors between physical qudits, and thus are often useful for fault tolerance, it is desirable to have as many transversal gates as possible. However the Eastin-Knill theorem \cite{eastinKnill} shows that a non-trivial ($ d \geq 2 $) code can have only finitely many transversal gates. So the best we can do is find codes whose transversal gates form a maximal finite subgroup. When a maximal group of transversal gates is achieved, for example the Clifford group, the next step is a fault tolerant implementation of a single gate $\tau$ outside of the group $\G$. The most popular method for implementing $\tau$ in a fault tolerant fashion is by using magic state distillation \cite{magicstatedist}, however, this is considered expensive \cite{magiccost1,magiccost2}. It is therefore crucial for fault tolerant gate synthesis to minimize the number of $\tau$ gates that are needed.

It was proven in \cite{superGoldenGates} that $\ico$, together with a particular $\tau$ gate, is the \textit{optimal} universal gate set for qubits in the sense that this gate set can quickly approximate any gate in $\SU(2)$ while minimizing the number of expensive $\tau$ gates that are needed. Therefore, our first example of $\ico$ transversal codes is well motivated.

Why is the $\ico$ universal gate set optimal? Unfortunately the work in \cite{superGoldenGates} is for the qubit case and there is no rigorous proof generalizing it to qudits (although some work towards $\SU(3)$ has already been accomplished \cite{goldenGatesPU(3)}). But one heuristic is that $\ico$ is a unitary $5$-group whereas every other subgroup of $\SU(2)$ is at most a unitary $3$-group. This means $\ico$ is more ``spread out" than the other subgroups and so an approximation algorithm doesn't have to use as many $\tau$ gates in order to reach the more hidden recesses of $\SU(2)$. 
The largest unitary $t$-group in $\SU(3)$ is in fact $\Sigma(360\phi)$. It is a unitary $3$-group whereas all other subgroups of $\SU(3)$ are at best unitary $2$-groups. Thus if we apply the same logic as before, we expect that $\Sigma(360\phi)$ plays a role in the optimal universal gate set for qutrits, but more work needs to be done to rigorously prove this claim. 

\emph{Disclaimer.--} The codes we find using our methods are in general non-additive, meaning they are not equivalent to any stabilizer code. In particular this means that the standard fault tolerant methods of decoding and correcting errors for stabilizer codes \cite{stabilizer} cannot be used. There are methods for decoding and correcting non-additive codes, for example, one can just measure the error space projectors and then undo the error with an appropriate unitary, but also more advanced techniques exist \cite{decoding1,decoding2,decoding3}. However, the jury is still out regarding how fault tolerant those non-additive decoding methods are (if at all). Our hope is that our work could be used to understand the structure of non-additive codes in more detail and help spur new research toward fully fault tolerant non-additive codes.

\section{Review of Unitary $t$-Groups}

Let $\U(q)$ be the unitary group of degree $q$. A finite subgroup $ \G $ of $ \U(q)$ is called a \textit{unitary $t$-group} \cite{designsHaarReview} if
\[
   \frac{1}{|\G|} \sum_{U \in \G} \qty( U \otimes U^*)^{\otimes t}  = \int_{\U(q)}  \qty(U \otimes U^*)^{\otimes t} \, dU, \numberthis  \label{eqn:ogdef}
\]
where the integral on the right is taken with respect to the unit-normalized Haar measure (if $ \G $ is merely a finite subset rather than a finite subgroup then this is called a unitary $ t $-design). 

On the right hand side, $U$ is a $q \times q$ matrix in the \underline{F}undamental (or defining) representation $\F$ of $\U(q)$. On the left hand side, $U$ is a $q \times q$ matrix in the restricted representation of $\F$ to $\G$, denoted by $\F^\downarrow$ or $\f$ (sometimes called the \textit{branching rules} in physics \cite{branching}). In our convention, we write representations with respect to $\U(q)$ using bold capital letters and we use the superscript $\downarrow$, or the corresponding bold lower-case letter, to denote the restriction to the finite subgroup $\G$. It follows that \cref{eqn:ogdef} is equivalent to
\[
   \frac{1}{|\G|} \sum_{g \in \G} \qty( \f \otimes \f^*)^{\otimes t}(g)  = \int_{\U(q)}  \qty(\F \otimes \F^*)^{\otimes t}(g) \, dg . \numberthis \label{eqn:repPOV}
\]
Here $\F^*$ is the dual representation of $\F$ given simply by the complex conjugate.

At this point we begin freely using concepts such as the character of a representation, inner product of characters, isotypic decomposition and isotypic projector, all of which are reviewed in the Supplemental Material \cite{supp}. 

Let $\irrep{1}$ denote the trivial irrep for both $\U(q)$ and $\G$. The character of this irrep is $1$ for all $g$. Thus we can multiply through by $ 1 $ on both sides to reveal that \cref{eqn:repPOV} is simply a projector equation:
\[
\Pi^{ (\f \otimes \f^*)^{\otimes t}}_{ \irrep{1} } = \Pi^{\qty( \F \otimes \F^*)^{\otimes t}}_{ \irrep{1} }. \numberthis
\]
That is, a unitary $t$-group is such that the projector of the $\U(q)$-representation $(\F \otimes \F^*)^{\otimes t}$ onto the trivial irrep $\irrep{1}$ must be the same as the projector of the $\G$-representation $(\f \otimes \f^*)^{\otimes t}$ onto the trivial irrep $\irrep{1}$. 

If we take the trace of both sides then we are counting the multiplicity of $\irrep{1}$ in the isotypic decomposition of the tensor product via an inner product of characters. Because $\f$ is a branched version of $\F$, and the trivial irrep cannot branch further, the multiplicity on the left is greater than or equal to the multiplicity on the right. That is, \textit{for any} subgroup $ \G $ of $ \U(q) $ we have
\[
    \expval*{1, (ff^*)^t} \geq  \expval*{1, (FF^*)^t}. \numberthis \label{eqn:easycomputation}
\]
And equality holds if and only if $\G$ is a unitary $t$-group. This is an inner product of characters where $ f $ and $ F $ denote the characters corresponding to the representations $ \f $ and $ \F $. We will continue to write the corresponding non-bold letters to represent characters.

Note that one can move characters within the inner product at the expense of a complex conjugation. Thus \cref{eqn:easycomputation} says $\expval{f^t,f^t} \geq \expval{F^t,F^t}$, or $ \lVert f^t \rVert \geq \lVert F^t \rVert $. Again equality holds if and only if $\G$ is a unitary $t$-group. Note that for a character $ f $ of a finite group it is standard to call $\expval{f,f}$ the \textit{norm}, rather than call it the norm squared, and to denote it by $ \lVert f \rVert:=\expval{f,f}$.

% what people call the norm of a character is actually its norm squared as a vector https://math.stackexchange.com/questions/4205910/must-the-norm-of-a-character-be-an-integer 

Now notice that $(\F \otimes \F^*)^{\otimes t}$ is a reducible $\U(q)$ representation and can be decomposed as
\[
    \qty( \F \otimes \F^*)^{\otimes t} = \bigoplus_{\irrE\in \mathcal{E}_t } (m_{\irrE} ) \, \irrE . \numberthis \label{eqn:ffdecomp}
\]
Here $\irrE$ is an ir\underline{R}ep of $ \U(q) $ and $\mathcal{E}_t$ is the set of all $ \U(q) $ irreps that appear with non-zero multiplicity $m_{\irrE}$.

As an example of the notation consider $\U(2)$. Then using \cite{LieART} we can compute: 
\begin{align}
    \qty( \irrep{2} \otimes \irrep{2}^*)^{\otimes 1} &= \irrep{1} \oplus  \irrep{3}, \\
    \qty( \irrep{2} \otimes \irrep{2}^*)^{\otimes 2} &= (2)\irrep{1} \oplus (3) \irrep{3} \oplus\irrep{5}, \\
    \qty( \irrep{2} \otimes \irrep{2}^*)^{\otimes 3} &= (5)\irrep{1} \oplus (9) \irrep{3} \oplus (5) \irrep{5} \oplus \irrep{7}.
\end{align}
Here (and elsewhere), parenthesis indicate multiplicity, for example $(2) \irrep{1}$ is shorthand for $\irrep{1} \oplus \irrep{1}$. Then we see that $\mathcal{E}_{1} = \{ \irrep{1}, \irrep{3} \},\mathcal{E}_{2} = \{ \irrep{1}, \irrep{3}, \irrep{5} \}, \mathcal{E}_{3} = \{ \irrep{1}, \irrep{3}, \irrep{5}, \irrep{7} \}$. More generally for $\U(2)$, we have $\mathcal{E}_t = \{ \irrep{1}, \irrep{3}, \irrep{5}, \cdots, \irrep{(2t+1)} \}$.

Returning to the general case $\U(q)$, we can take \cref{eqn:repPOV} and insert the decomposition from \cref{eqn:ffdecomp} to obtain
\[
  \bigoplus_{\irrE\in \mathcal{E}_t } (m_{\irrE}) \, \frac{1}{|\G|} \sum_{g \in \G} \irre(g) = \bigoplus_{\irrE\in \mathcal{E}_t } (m_{\irrE}) \, \int_{\U(q)} \irrE(g) \, dg. \numberthis
\]
Here $\irre$ denotes the restriction of $\irrE$ to $\G$. Thus we see that $\G$ is a unitary $t$-group if and only if
\[
   \frac{1}{|\G|} \sum_{g \in \G} \irre (g) =  \int_{\U(q)} \irrE(g) \, dg \qquad \forall \irrE\in \mathcal{E}_t.\numberthis
\]
Inserting the trivial character $1$ into both sides, we see that this is an equality of projectors 
 $ \Pi^{\irre}_{\irrep{1} } = \Pi^{\irrE}_{\irrep{1} }$ for all $ \irrE\in \mathcal{E}_t.$ However, notice that when $\irrE= \irrep{1}$ this equality is trivially satisfied (because $ \irrep{1}^\downarrow=\irrep{1}$). On the other hand, when $\irrE\neq \irrep{1}$ the right hand side is the zero-matrix $\bm{0}$, that is,
\[
    \Pi^{\irre}_{\irrep{1} } = \bm{0} \qquad \forall \irrE\in \mathcal{E}_t : \irrE\neq \irrep{1}. \numberthis
\]
In a similar fashion as before we can take the trace of both sides to get that \textit{for any} $\G$,
\[
    \expval{1, \irrechar } \geq 0 \quad \forall \irrE\in \mathcal{E}_t : \irrE\neq \irrep{1}, \numberthis
\]
with equality if and only if $\G$ is a unitary $t$-group. As before $\irrechar$ is the character corresponding to the representation $\irre$.

A summary of all the equivalent conditions derived above can be found in the Supplemental Material \cite{supp} but for our purposes, the most useful perspectives on unitary $ t $-groups are the two we highlight below.

\begin{lemma}\label{lem:unitarytgroup1} $\G \subset \U(q)$ is a unitary $t$-group if and only if either of the following equivalent conditions are satisfied:
\begin{enumerate}[(1)]
    \item $ \lVert f^{ t} \rVert = \lVert F^{ t} \rVert $,
    \item $\expval{1,\irrechar} = 0,  \quad \forall \irrE\in \mathcal{E}_t, \irrE\neq \irrep{1}$.
\end{enumerate}
\end{lemma}

For small $ t $, we can compute $ \lVert F^t \rVert $ to obtain even simpler criteria for identifying a unitary $ t $-group (c.f. \cite{GrossCharacters&tGroups}).

\begin{lemma}\label{lem:unitarytgroup2}
    Suppose $\G \subset \U(q)$. 
    \begin{enumerate}[(1)]
        \item $\G$ is a unitary $1$-group $\iff$ $ \lVert f \rVert = 1$ (i.e., $\f$ is irreducible),
        \item  $\G$ is a unitary $2$-group $\iff$ $ \lVert f^2 \rVert = 2$.
    \end{enumerate}
\end{lemma}

\begin{proof}
The fundamental representation $ \F $ of $ \U(q) $ is irreducible, so $  \lVert F \rVert =1 $. It is well known that when $\F$ is the fundamental representation of $\U(q)$, then $ \F \otimes \F^* = \irrep{1} \oplus \bmsf{Ad} $ where $\bmsf{Ad}$ is the adjoint irrep of $\SU(q)$. Then $ \lVert F^2 \rVert = \lVert FF^* \rVert $ which can be evaluated as $ \lVert 1+Ad \rVert = 2 $ where $Ad$ is the adjoint character. Note that we have used the irreducibility of $\bmsf{Ad}$ which is equivalent to the fact that $\SU(q)$ is a simple Lie group.
\end{proof}

Going further, $ \G \subset \U(q) $ is a $ 3 $-group if and only if   $ \lVert f^3 \rVert =6 $, for $ q \geq 3 $, or  $ \lVert f^3 \rVert =5 $, for $ q=2 $ \cite{GrossCharacters&tGroups}. The classification of unitary $ t $-groups given in \cite{tGroupsBannai} shows that no $ t $-groups exist for $ t >3 $, with the exception of the binary icosahedral subgroup $\mathrm{2I} $ of $ \U(2)$, which is the only unitary $4$-group and the only unitary $5$-group.

\section{Twisted Unitary $t$-groups}

In \cref{eqn:repPOV}, the left hand side can be thought of as a sum with respect to the uniform weight $1/|\G|$. But more generally we can take the sum $\sum_{g \in \G} W(g) \qty(\f \otimes \f^*)^{\otimes t}(g)$ with respect to any normalized weight $W$. If
\[
    \sum_{g \in \G} W(g) \qty( \f \otimes \f^*)^{\otimes t}(g)  = \int_{\U(q)}  \qty(\F \otimes \F^*)^{\otimes t}(g) \, dg , \numberthis \label{eqn:weightedtGroup}
\]
it is standard to call $\G$ a \textit{weighted $t$-group} \cite{unitary_designs_and_codes}.

Let $\logirr$ be an irrep of $\G$ with corresponding character $\lambda$. Then there is a natural weight
\[
W_{\logirr} (g) =  \frac{1}{|\G| } \lambda^*(g) \lambda(g) = \frac{1}{|\G| } |\lambda (g)|^2. \numberthis
\]
Note that because $\logirr$ is irreducible then $W_{\logirr}$ is normalized: $\sum_{g \in \G} W_{\logirr}(g) =\lVert \lambda \rVert = 1 $. And when $\logirr$ is $1$-dimensional then $W_{\logirr}(g) = 1/|\G|$ and so we recover the usual (unweighted) definition of unitary $t$-groups. When we use the weight $W_{\logirr}$ we will call $\G$ a \textit{twisted unitary $t$-group} with respect to $\logirr$ or, equivalently, a $\logirr$-twisted unitary $t$-group. We can adapt \cref{lem:unitarytgroup1} to this scenario, for a proof see \cite{supp}.

\begin{lemma}\label{lem:twistedunitarytgroup1} $\G \subset \U(q)$ is a $ \logirr $-twisted unitary $t$-group if and only if either of the following equivalent conditions are satisfied:
\begin{enumerate}[(1)]
    \item $ \lVert \lambda f^{ t} \rVert= \lVert F^{ t} \rVert $,
    \item $\expval{\lambda^* \lambda,\irrechar} = 0,  \quad \forall \irrE\in \mathcal{E}_t, \irrE\neq \irrep{1}$.
\end{enumerate}
\end{lemma}

Notice that part (2) of \cref{lem:twistedunitarytgroup1} implies part (2) of \cref{lem:unitarytgroup1}, so every $ \logirr $-twisted unitary $ t $-group is also a unitary $ t $-group. To see this implication, it is enough to observe the tensor product of an irrep with its dual always contains a unique trivial subrepresentation, i.e.,  there is a unique way to write $ \logirr^* \otimes \logirr = \irrep{1} \oplus \bmsf{\xi} $ for some  representation $\bmsf{\xi}$. Thus $\expval{\lambda^* \lambda,\irrechar}=\expval{1,\irrechar}+ \expval{\xi,\irrechar} $ and so $ \expval{\lambda^* \lambda,\irrechar} = 0 $ implies $\expval{1,\irrechar} = 0 $.

We can also adapt \cref{lem:unitarytgroup2} to the $ \logirr $-twisted case.

\begin{lemma}\label{lem:twistedunitarytgroup2}
    Suppose $\G \subset \U(q)$. Let $ \logirr $ be an irrep of $ \G $.
    \begin{enumerate}[(1)]
        \item $\G$ is a $ \logirr $-twisted unitary $1$-group $\iff$  $ \lVert \lambda f \rVert = 1$ (i.e., $ \logirr \otimes \f$ is irreducible),
        \item $\G$ is a $ \logirr $-twisted unitary $2$-group $\iff$ $ \lVert \lambda f^2 \rVert = 2$.
    \end{enumerate}
\end{lemma}

This lemma is the reason we have dubbed these ``twisted" unitary $t$-groups. The irrep $\logirr$ latches onto the fundamental irrep $\f$ (or powers thereof) and \textit{twists}. When $\logirr$ is the trivial irrep $\irrep{1}$ then no such twisting occurs and we reproduce the usual concept of unitary $t$-groups. A similar phenomenon can be found in twisted Gelfand pairs \cite{twisted1} and twisted wavefunctions from induced representations \cite{twisted2}.

\section{Quantum Codes}

The Knill-Laflamme (KL) conditions \cite{KL} state that a quantum code has distance $d = t + 1$ iff
\[
    \bra{\psi} E \ket{\phi} = c_E \braket{\psi}{\phi}, \numberthis
\]
for all errors $E$ of weight $t$ or less and all codewords $ \ket{\psi}, \ket{\phi} $. In other words, we can correct $\lfloor t/2 \rfloor$ errors or detect $t$ errors. \textit{Achtung}! Usually $t$ denotes the number of errors that can be \textit{corrected}, but in order for us to preserve the $t$ in unitary $t$-design we had to resort to a convention where $t$ denotes the number of errors that can be \textit{detected}.

Consider a finite subgroup $\G$ of $\U(q)$. Suppose all codewords $\ket{\psi}$ transform in an irrep $\logirr$ of $\G$ (where we think of \underline{L}ambda as the irrep of the \underline{L}ogical codespace). In particular the codespace will have dimension $ |\logirr| $, the dimension of the irrep $ \logirr$. And suppose some error $E$ transforms in a representation $\irrE$ of $\U(q)$, which branches to the representation $ \irre $ of $ \G $. Then the tensor product $\bra{\psi} \otimes E \otimes \ket{\phi}$ transforms in the representation $\logirr^* \otimes \irre  \otimes \logirr$ of $\G$. The contraction of any tensor is always an invariant, thus $\bra{\psi} E \ket{\phi}$ is an invariant and so must transform in the trivial representation $\irrep{1}$ or the null representation $\bm{0}$.

But if $\logirr^* \otimes \irre  \otimes \logirr$ does not contain a copy of the trivial representation $\irrep{1}$ in its isotypic decomposition (i.e., $\expval{1, \lambda^* \irrechar \lambda} = 0$), then the only option is that $\bra{\psi} E \ket{\phi} = 0$ and thus the KL condition is satisfied for the error $E$. We have proven the following lemma. 

\begin{lemma} \label{lem:codesnontrivial} Suppose a code transforms in an irrep $\logirr$ of $\G$ and an error $E$ transforms in an irrep $\irrE$ of $\U(q)$. If $\expval{1,\lambda^* \irrechar \lambda} = 0$ then the KL conditions are automatically satisfied for the error $E$. 
\end{lemma}

Notice that \cref{lem:codesnontrivial} does not apply to $\irrE=\irrep{1}$ because $\expval{1, \lambda^*1\lambda}= \lVert \lambda \rVert \neq 0 $.  However, in the $\irrE=\irrep{1}$ case the KL conditions are also automatically satisfied. The idea is that if $E$ transforms trivially with respect to $\G$ then $ E $ commutes with the action of $ \G $ and thus, by Schur's lemma, $ E $ acts proportional to the identity on irreps of $\G$, i.e., $E \ket{\phi} = c_E \ket{\phi}$. Then $\bra{\psi} E \ket{\phi} = c_E \braket{\psi}{\phi}$. Note that here the KL condition may be satisfied in a degenerate manner, $ c_E \neq 0 $, whereas for the errors in \cref{lem:codesnontrivial} we always have $ c_E=0 $. This proves the following lemma.

\begin{lemma}\label{lem:codestrivial} If a code transforms in an irrep $\logirr$ of $\G$ and if an error $E$ transforms in the trivial irrep $\irrep{1}$ then the KL conditions are automatically satisfied for the error $E$. 
\end{lemma}

Using these two lemmas we derive our main result.

\begin{theorem}\label{thm:APP} 
If $\G$ is a $\logirr$-twisted unitary $t$-group then every subspace of $\f^{\otimes n}$ that transforms in $\logirr$ is a $|\logirr|$-dimensional quantum code with distance $d\geq t+1$ and transversal gate group $\Glog = \logirr(\G)$. 
\end{theorem}

\begin{proof}
Recall that $\mathcal{E}_t$ is the set of all irreps in the isotypic decomposition of $(\F \otimes \F^*)^{\otimes t}$. So any error $ E $ of weight $ t $ or less can be decomposed as $E= \sum_{\irrE \in \mathcal{E}_t } E_{\irrE} $ where $ E_{\irrE} $ is $ E $ projected onto the $ \irrE $-isotypic subspace of $ (\F \otimes \F^{*})^{\otimes t} $. 
By \cref{lem:twistedunitarytgroup1} we have that $\expval{\lambda^*   \lambda, \irrechar } = 0$, and equivalently $\expval{1,\lambda^* \irrechar  \lambda} = 0$, for every $\irrE \in \mathcal{E}_t, \irrE \neq \irrep{1}$. 
So we can apply \cref{lem:codesnontrivial} to conclude that $ \bra{\psi} E_{\irrE} \ket{\phi}=0 $ for every $\irrE \in \mathcal{E}_t, \irrE \neq \irrep{1}$. We are left with 
\[
\bra{\psi} E \ket{\phi}= \sum_{\irrE \in \mathcal{E}_t } \bra{\psi}  E_{\irrE} \ket{\phi}   =\bra{\psi} E_{\irrep{1}} \ket{\phi}=c_{E_{\irrep{1}}} \braket{\psi}{\phi}, \numberthis
\]
where the final equality follows from \cref{lem:codestrivial}.

 When $ \logirr $ is a subrepresentation of $ \f^{\otimes n} $ then the transversal gate $\f^{\otimes n}(g)=g^{\otimes n}$ is a logical gate implementing $\logirr(g)$ on the codespace. Thus all gates from $\Glog:= \logirr(\G)$, the image of the representation $ \logirr $, can be implemented transversally. 
\end{proof}

In the two examples below, the way that we use this theorem is via \cref{lem:twistedunitarytgroup2}. That is, we simply check if $ \lVert \lambda f \rVert =1$ or $ \lVert \lambda f^2 \rVert = 2$ to see if $\G$ is a $\logirr$-twisted unitary $1$-group or a $\logirr$-twisted unitary $2$-group respectively. For higher orders of $t$, one must invoke \cref{lem:twistedunitarytgroup1} and use branching rules $\U(q) \downarrow \G$ (see the Supplemental Material \cite{supp}).

A good heuristic when applying this theorem is to pick a ``large" $\G$ inside of $\U(q)$ and pick a $\logirr$ which is either faithful or almost faithful (i.e., the kernel of the representation $\logirr$ should be small). That way the image $\Glog=\logirr(\G)$ is a very large transversal gate group. From this perspective, \cref{thm:APP} is a method to construct \textit{designer} quantum codes having a certain transversal gate group (c.f. \cite{us2}).

\section{Example 1: $\ico$ Qubit Codes} Let us apply our theorem to the binary icosahedral group $\G = \ico$ in $\U(2)$, which is a unitary $5$-group. The character table can be found in the Supplemental Material \cite{supp} which is taken directly from GAP \cite{GAP} as \texttt{PerfectGroup(120)}. GAP labels the irreps as $\irrGAP_i$ where $i$ ranges between $1$ and $9$. There are two $2$-dimensional irreps, $\irrGAP_2$ and $\irrGAP_3$, and we will take $\f = \irrGAP_2$ as our fundamental irrep. 

Suppose our code transforms in the other $2$-dimensional irrep: $\logirr = \irrGAP_3$. Then one can check in GAP that  $ \lVert \lambda f \rVert = 1$ and $ \lVert \lambda f^2 \rVert = 2$ (we provide a code snippet in the Supplemental Material \cite{supp}). Using \cref{lem:twistedunitarytgroup2} we see that $\ico$ is a $\irrGAP_3$-twisted unitary $2$-group. So by \cref{thm:APP}, any subspace that transforms in $\irrGAP_3$ will be a code with distance $d = 3$ and will implement $\ico$ transversally. This supersedes the codes found in \cite{us1}.

As a canonical example, suppose we encode a qubit, transforming in $\irrGAP_3$, into an $n$-th tensor power of qubits, transforming in $\irrGAP_2^{\otimes n}$. In order for a code to be present, we need the multiplicity of $\irrGAP_3$ in $\irrGAP_2^{\otimes n}$ to be at least 1, i.e., we need to find $n$ such that $\expval{\irrGAPchar_3, \irrGAPchar_2^n} > 0$. The smallest code occurs when $n = 7$ and the multiplicity is $\expval{\irrGAPchar_3, (\irrGAPchar_2)^7}=1$, meaning that this code is unique (for more on unique codes see the Supplemental Material \cite{supp}). Since $\irrGAP_3$ only occurs in odd tensor powers the next smallest code occurs when $n = 9$, and the multiplicity is $8$. This means there is a $\mathbb{C}P^7$-moduli space worth of non-equivalent codes (see Supplemental Material \cite{supp}). Going further, there is a $\mathbb{C}P^{43}$-moduli space of codes in $n = 11$ qubits and a $\mathbb{C}P^{208}$-moduli space of codes in $n = 13$ qubits. In fact there are $\irrGAP_3$ codes for all odd $n \geq 7$.

\section{Example 2: $\Sigma(360\phi)$ Qutrit Codes}

As another example consider $\G = \Sigma(360\phi)$ in $\U(3)$, a unitary 3-group that appears in the high-energy physics literature \cite{SU3Fairbairn1964FiniteAD,SU3Ludl_2011,SUtree}. The character table, a GAP code snippet, and branching rules can be found in the Supplemental Material \cite{supp}. The irrep $\irrGAP_1 = \irrep{1}$ is the trivial irrep and there are four different 3-dimensional irreps, labeled $\irrGAP_2$, $\irrGAP_3$, $\irrGAP_4$, and $\irrGAP_5$. We will take the fundamental irrep to be $\f = \irrGAP_2$. 

Suppose our code transforms as $\logirr = \irrGAP_3$. Then one can check in GAP that $ \lVert  \irrGAP_3 \irrGAP_2 \rVert = 1$. Thus by \cref{lem:twistedunitarytgroup2}, $\Sigma(360\phi)$ is a $\irrGAP_3$-twisted unitary $1$-group. So by \cref{thm:APP}, any subspace that transforms in $\irrGAP_3$ will be a code with distance $d = 2$ and will implement $\Sigma(360\phi)$ transversally. Let us find a qutrit code transforming in $\irrGAP_3$ within a tensor product of $n$ qutrits $\f^{\otimes n}$. Again we simply look for $n$ such that $\expval{\irrGAPchar_3,\irrGAPchar_2^n} > 0$. Thus we have proven that there are codes in $n = 7, 10, 13, 16, 19, \cdots$, i.e., whenever $n \equiv 1 \mod{3}$ for $n \geq 7$. The smallest code transforming in $ \irrGAP_3 $ encodes 1 qutrit into $7$ qutrits, detects any single error, and transversally implements any gate from $\Sigma(360\phi)$.  

However, unlike $\ico$, for $\Sigma(360\phi)$ there is also another good logical irrep. Let $\logirr = \irrGAP_4$. One can check that $\Sigma(360\phi)$ is a $\irrGAP_4$-twisted unitary 1-group and there are codes whenever $n \equiv 2 \mod{3}$ and $n \geq 5$. The smallest code here is actually better, it occurs when $n = 5$ and the multiplicity is $\expval{\irrGAPchar_4, (\irrGAPchar_2)^5}=1$, meaning that this code is unique (for more on the history of this code \cite{sigma360Ian,sigma360Victor} and unique codes in general see the Supplemental Material \cite{supp}). All the codes in this family encode 1 qutrit into $ n $ qutrits and detect any single error while implementing $\Sigma(360\phi)$ transversally.

\section{Conclusion}

This paper establishes a novel and significant connection between quantum $t$-designs, specifically twisted unitary $t$-groups, and quantum error-correcting codes. By introducing twisted unitary $t$-groups, which generalize unitary $t$-groups through the incorporation of irreducible representations, we provide a framework for constructing quantum codes with many transversal gates, which naturally do not spread errors and thus are useful for fault tolerance.  Two illustrative examples involving the unitary $5$-group $ \ico $ in $ \SU(2) $ and the unitary 3-group $\Sigma(360\phi)$ in $ \SU(3) $ highlight the practicality and versatility of our approach, yielding $n$-qubit and $n$-qutrit quantum codes with impressive transversal gates. Both of these codes have transversal gate groups which are maximal, lacking only a single gate outside of the respective groups $ \ico $ and $\Sigma(360\phi)$ to achieve universality. 

It is the hope of the authors that this work, on top of previous work on quantum error correcting codes outside the stabilizer framework \cite{nonadditive,nonadditive2,nonadditive3,perminvOuyang,barg,us1,us2}, will spur a robust inquiry into quantum circuits to implement error correction, fault tolerant measurements, fault tolerant gates, and general fault tolerant circuit design, all for nonadditive codes.

\section{Acknowledgments} 

We thank Markus Heinrich for first introducing us to unitary $t$-groups and we thank Michael Gullans and Victor V. Albert for helpful conversations regarding unitary $t$-designs and code finding. We thank Anthony Leverrier for thoroughly reviewing a previous version of the manuscript and finding an error. We also thank Brad Horner from Mathematics Stack Exchange for suggesting the proof in the last section of the Supplemental Material. This research was supported in part by the MathQuantum RTG through the NSF RTG grant DMS-2231533.

\nocite{2004permutation,gross1}
\bibliography{biblio}

\appendix 

\makeatletter
\renewcommand{\theequation}{S\arabic{equation}}
\renewcommand{\thetable}{S\arabic{table}}
\renewcommand{\thefigure}{S\arabic{figure}}
\renewcommand{\thelemma}{S\arabic{lemma}}
\renewcommand{\thetheorem}{S\arabic{theorem}}
\setcounter{table}{0}
\setcounter{figure}{0}
\setcounter{lemma}{0}
\setcounter{equation}{0}

\newpage
\onecolumngrid
\section{Supplemental Material}
\twocolumngrid

\section{Equivalent Conditions for Unitary $t$-groups}

\begin{lemma} The following are equivalent:
\begin{enumerate}[(1)]
    \item $\G \subset \U(q)$ is a unitary $t$-group
    \item $\frac{1}{|\G|} \sum_{g \in \G} \qty( \f \otimes \f^*)^{\otimes t}(g)  = \int_{\U(q)}  \qty(\F \otimes \F^*)^{\otimes t}(g) \, dg$
    \item $\Pi^{\qty( \f \otimes \f^*)^{\otimes t}}_{ \irrep{1} } = \Pi^{\qty( \F \otimes \F^*)^{\otimes t}}_{ \irrep{1} }$
    \item $\expval*{1,(ff^*)^{t} } = \expval*{1,(FF^*)^{t}}$
    \item $ \lVert f^{ t} \rVert= \lVert F^{ t} \rVert $
    \item $\frac{1}{|\G|} \sum_{g \in \G} \irre(g) =  \int_{\U(q)} \irrE(g) \, dg, \qquad \forall \irrE\in \mathcal{E}_t$ 
    \item $\frac{1}{|\G|} \sum_{g \in \G} \irre(g) = \bm{0},  \quad \forall \irrE\in \mathcal{E}_t, \irrE\neq \irrep{1} $ 
    \item $\Pi^{\irre}_{ \irrep{1}} = \bm{0},  \quad \forall \irrE\in \mathcal{E}_t, \irrE\neq \irrep{1}$ 
    \item $\expval{1,\irrechar} = 0,  \quad \forall \irrE\in \mathcal{E}_t, \irrE\neq \irrep{1}$ 
\end{enumerate}
\end{lemma}

\section{Equivalent Conditions for \textit{Twisted} Unitary $t$-groups}

\begin{lemma} The following are equivalent:
\begin{enumerate}[(1)]
    \item $\G \subset \U(q)$ is a $\logirr$-twisted unitary $t$-group
    \item $ \frac{1}{|\G| } \sum\limits_{g \in \G} |\lambda(g)|^2  \qty( \f \otimes \f^*)^{\otimes t}(g)  = \int\limits_{\U(q)}  \qty(\F \otimes \F^*)^{\otimes t}(g) \, dg$
    \item $\Pi^{\qty( \f \otimes \f^*)^{\otimes t}}_{ \logirr^* \otimes \logirr } = \Pi^{\qty( \F \otimes \F^*)^{\otimes t}}_{ \irrep{1} }$
    \item $\expval*{\lambda^* \lambda , \qty(f  f^*)^{ t} } = \expval*{1 , (F F^*)^{t} } $
    \item $ \lVert \lambda f^{ t} \rVert = \lVert F^t \rVert $
    \item $\frac{1}{|\G| } \sum_{g \in \G}  |\lambda(g)|^2 \, \irre(g) = \int_{\U(q)} \irrE(g) \, dg, \quad \forall \irrE\in \mathcal{E}_t$
    \item $\frac{1}{|\G| } \sum_{g \in \G} |\lambda(g)|^2 \, \irre(g) = \bm{0},  \quad \forall \irrE\in \mathcal{E}_t, \irrE\neq \irrep{1}$ 
    \item $\Pi^{\irre}_{\logirr^* \otimes \logirr } = \bm{0},  \quad \forall \irrE\in \mathcal{E}_t, \irrE\neq \irrep{1}$ 
    \item $\expval{\lambda^* \lambda, \irrechar   } = 0,  \quad \forall \irrE\in \mathcal{E}_t, \irrE\neq \irrep{1}$ 
\end{enumerate}
\end{lemma}

\section{Character Theory Background}

Here we review character theory and other relevant background in a manner similar to \cite{GrossCharacters&tGroups}. Every representation $ \irrIco $ of a finite group $ \G $ has a corresponding character $ \irrIcochar(g)=\text{tr}( \irrIco(g)) $, where $ \text{tr} $ denotes the trace of the matrix. The scalar product between characters is given by
\[
\expval{\irrIcochar,\irrGAPchar} := \frac{1}{|\G|} \sum_{g \in \G}  \irrIcochar^*(g) \chi(g). \numberthis
\]
Denote the irreducible representations
(irreps) of $ \G $ by $  \irrGAP_i $ and their associated irreducible characters by $ \irrGAPchar_i $. 
It is a fundamental relation that the irreducible characters are orthonormal $ \expval{\irrGAPchar_i,\irrGAPchar_j}=\delta_{ij} $. The fact that any representation $ \irrIco $ of a finite group reduces to a direct sum of irreps $ \irrIco= \bigoplus_i (m_i) \irrGAP_i $
means that any character can be expanded in terms of the irreducible ones and further that $ \expval{\irrIcochar,\irrGAPchar_i}=m_i $ gives the number of
times
that $ \irrGAP_i $ occurs in the decomposition of $ \irrIco $. The norm of a character is defined to be $ \lVert \irrIcochar \rVert :=\expval{\irrIcochar,\irrIcochar}=\sum_i m_i^2 $, which is always a positive integer (indeed the scalar product of any two characters is always a nonnegative integer), and moreover $ \lVert \irrIcochar \rVert =1 $ if and only if $ \irrIco $ is irreducible. 

The direct sum $ \irrIco= \bigoplus_i (m_i) \irrGAP_i $ is called the \textit{isotypic decomposition} of $ \irrIco $.
The subspace
\[
(m_i) \irrGAP_i := \underbrace{\irrGAP_i \oplus \cdots \oplus \irrGAP_i}_{m_i} \numberthis
\]
of $ \irrIco $ is canonical and is called the $ \irrGAP_i $-isotypic subspace of $ \irrIco $. The projector onto this space is given by 
\[
\Pi^{\irrIco}_{\irrGAP_i}= \frac{|\irrGAP_i|}{|\G|} \sum_{g \in \G} \irrGAP_i^*(g) \irrIco(g),\numberthis
\]
where $ |\irrGAP_i| $ is the dimension of the representation $ \irrGAP_i $ (also called the degree). For example if $ \irrGAP_i $ is the trivial irrep $ \irrep{1} $ the projector is
\[
\Pi^{\irrIco}_{\irrep{1}}= \frac{1}{|\G|} \sum_{g \in \G} \irrIco(g). \numberthis
\]

\section{Proof of \cref{lem:twistedunitarytgroup1}}

The proof of \cref{lem:twistedunitarytgroup1} is similar to the proof of \cref{lem:unitarytgroup1}, but we present it here for completeness.

Recall that we defined $ \G $ to be a $ \logirr $-twisted unitary $ t $-group if
\[
    \tfrac{1}{|\G|} \sum_{g \in \G} |\lambda(g)|^2 \qty( \f \otimes \f^*)^{\otimes t}(g)  = \int_{\U(q)}  \qty(\F \otimes \F^*)^{\otimes t}(g) \, dg . \numberthis \label{eqn:twistedtGroup}
\]
Since $ |\lambda(g)|^2= \lambda^*(g)\lambda(g) $ is the character of $ \logirr^* \otimes \logirr $, we can recognize \cref{eqn:twistedtGroup} as simply a projector equation:
\[
\Pi^{ (\f \otimes \f^*)^{\otimes t}}_{ \logirr^* \otimes \logirr } = \Pi^{\qty( \F \otimes \F^*)^{\otimes t}}_{ \irrep{1} }. \numberthis
\]
That is, a $ \logirr $-twisted unitary $ t $-group is such that the projector of the $\U(q)$-representation $(\F \otimes \F^*)^{\otimes t}$ onto the trivial irrep $\irrep{1}$ must be the same as the projector of the $\G$-representation $(\f \otimes \f^*)^{\otimes t}$ onto the isotypic subspaces for the $\G$-representation $ \logirr^* \otimes \logirr $. 

If we take the trace of both sides of \cref{eqn:twistedtGroup}  then we arrive at the equation
\[
    \expval*{\lambda^* \lambda, (ff^*)^t} =  \expval*{1, (FF^*)^t}. \numberthis \label{eqn:easycomputationtwisted}
\]

Note that one can move characters within the inner product at the expense of a complex conjugation. Thus \cref{eqn:easycomputationtwisted} says $\expval{\lambda f^t,f^t} = \expval{F^t,F^t}$, or $ \lVert \lambda f^t \rVert = \lVert F^t \rVert $.

Now notice that $(\F \otimes \F^*)^{\otimes t}$ is a reducible $\U(q)$ representation and can be decomposed as
\[
    \qty( \F \otimes \F^*)^{\otimes t} = \bigoplus_{\irrE\in \mathcal{E}_t } (m_{\irrE} ) \, \irrE . \numberthis \label{eqn:ffdecomp}
\]
So we can take \cref{eqn:twistedtGroup} and insert the decomposition from \cref{eqn:ffdecomp} to obtain
\[
  \bigoplus_{\irrE\in \mathcal{E}_t } (m_{\irrE}) \, \frac{1}{|\G|} \sum_{g \in \G} |\lambda(g)|^2 \, \irre(g) = \bigoplus_{\irrE\in \mathcal{E}_t } (m_{\irrE}) \, \int_{\U(q)} \irrE(g) \, dg. \numberthis
\]
Here $\irre$ denotes the restriction of $\irrE$ to $\G$. Thus we see that $\G$ is a unitary $t$-group if and only if
\[
   \frac{1}{|\G|} \sum_{g \in \G} |\lambda(g)|^2 \, \irre (g) =  \int_{\U(q)} \irrE(g) \, dg \qquad \forall \irrE\in \mathcal{E}_t.\numberthis
\]
Which we can recognize as an equality of projectors 
 $ \Pi^{\irre}_{\logirr^* \otimes \logirr } = \Pi^{\irrE}_{\irrep{1} }$ for all $ \irrE\in \mathcal{E}_t.$ However, notice that when $\irrE= \irrep{1}$ this equality is trivially satisfied (because $ \logirr^* \otimes \logirr $ contains a unique copy of the trivial irrep). On the other hand, when $\irrE\neq \irrep{1}$ the right hand side is the zero-matrix $\bm{0}$, that is,
\[
    \Pi^{\irre}_{\logirr^* \otimes \logirr} = \bm{0} \qquad \forall \irrE\in \mathcal{E}_t : \irrE\neq \irrep{1}. \numberthis
\]
In a similar fashion as before we can take the trace of both sides to get the final condition
\[
    \expval{\lambda^* \lambda, \irrechar } = 0 \quad \forall \irrE\in \mathcal{E}_t : \irrE\neq \irrep{1}. \numberthis
\]

\section{Unique Codes and Code Multiplicity}

The $ 7 $ qubit $ \ico $ code is unique in the sense that all $ 7 $ qubit error correcting codes with transversal gate group $ \ico $ must be equivalent via non-entangling gates. This code was first written down in \cite{2004permutation} in the process of investigating small permutationally invariant codes. Later, a spin code corresponding to the $ 7 $ qubit $ \ico $ code was independently discovered in \cite{gross1}. Finally, a previous paper of the authors rediscovered the $ 7 $ qubit $ \ico $ code, and was the first work to determine the transversal gates of the code and make connections with fault tolerance \cite{us1}. That work also proved the uniqueness of the $ 7 $ qubit $ \ico $ code by observing that a certain irrep occured with multiplicity $ 1 $. 

Inspired by the $ 7 $ qubit $ \ico $ code, one of the authors (IT) conjectured that all codes corresponding to multiplicity $ 1 $ irreps always have good distance ($ d\geq 2$), and used this technique to discover the unique $ 5 $ qutrit $ \Sigma(360 \phi) $ code, which occurs in this paper as the smallest code in the second example code family given in the main text \cite{sigma360Ian}. This technique of searching for multiplicity $ 1 $ codes was communicated to Victor V. Albert, who used it to independently rediscover the unique $ 5 $ qutrit $ \Sigma(360 \phi) $  code \cite{sigma360Victor}. Although unique objects are always of special interest, further investigation revealed that nearly all multiplicity $ 1 $ codes have trivial distance $ d=1 $ and thus a completely different approach to understanding the $ 7 $ qubit $ \ico $ code and the $ 5 $ qutrit $ \Sigma(360 \phi) $ code had to be undertaken.

Far from being unique, most of the codes in this paper live in a continuously varying moduli space of non-equivalent codes with identical code parameters and transversal gate group. Indeed, if $ \G $ is a $ \logirr $ twisted unitary $ t $-group then the moduli space of code constructed from \cref{thm:APP} is exactly the complex projective space $ \mathbb{C}P^{m_{\logirr}-1} $ (note that if the multiplicity $ m_{\logirr} $ is $ 1 $, as in the special cases discussed above, then the moduli space of codes degenerates to  $ \mathbb{C}P^{0} $, a single point, thus proving uniqueness).

To rigorously confirm that the moduli space is $ \mathbb{C}P^{m_{\logirr}-1} $, we start with the isotypic decomposition $ \bigoplus_i (m_i) \irrGAP_i $ of $ \f^{\otimes n} $. By Schur's lemma, $ \mathrm{Hom}_{\G}(\logirr,  \irrGAP_i )=0 $ for $ \irrGAP_i \neq \logirr $ while $ \mathrm{Hom}_{\G}(\logirr,  \logirr )=\mathbb{C} $ (where $\mathrm{Hom}_{\G}(\logirr,  \irrGAP_i )$ is the set of $\G$-equivariant linear maps from $\logirr$ to $\irrGAP_i$ and $\mathbb{C}$ is the complex numbers transforming in the trivial irrep of $ \G $). Thus 
\begin{align}
    \mathrm{Hom}_{\G}(\logirr, \f^{\otimes n} ) &= \mathrm{Hom}_{\G}(\logirr, \bigoplus_i (m_i) \irrGAP_i )\\
    &= \bigoplus_i (m_i) \mathrm{Hom}_{\G}(\logirr,  \irrGAP_i )\\
   &= (m_{\logirr})\mathbb{C}\\
    &= {\mathbb{C}}^{m_{\logirr}}.
\end{align}
Every $ \logirr $ subrepresentation of $ \f^{\otimes n} $ is the image of a nonzero $\G$-equivariant linear map from $ \logirr $ to $ \f^{\otimes n} $. By the calculation above we see that the space of nonzero $\G$-equivariant linear maps is $ \mathbb{C}^{m_{\logirr}} \setminus \{0\} $. We can apply Schur' lemma again to conclude that two such maps will have the same image if and only if they are nonzero scalar multiples. Thus the moduli space of $ \logirr $ codes in $ \f^{\otimes n} $ is $ \mathbb{C}^{m_{\logirr}} \setminus \{0\}  $ modulo $ \mathbb{C}^{\times} $, which is  exactly the complex projective space $ \mathbb{C} {P}^{{m_{\logirr}}-1}$.

\newpage 
\onecolumngrid
\section{ {\large Example 1: $\ico$ Qubit Codes} }

% &   \text{Order} & 1 & 5 & 5 & 3 & 4 & 6 & 10 & 2 & 10 \\ 
\onecolumngrid
\begin{table}[htp]
    \centering
    $$
    \begin{array}{cc|ccccccccc} \toprule
 &  \text{\cite{us1}} & 1 & 12 & 12 & 20 & 30 & 20 & 12 & 1  & 12 \\ \midrule
 \irrGAP_1 & \irrIco_1 & 1 & 1 & 1 & 1 & 1 & 1 & 1 & 1 & 1 \\
 \irrGAP_2 &  \irrIco_2 & 2 & \zeta_5^4+\zeta_5 & \zeta_5^3+\zeta_5^2 & -1 & 0 & 1 & -\zeta_5^4-\zeta_5 & -2 & -\zeta_5^3-\zeta_5^2 \\
\irrGAP_3 & \overline{\irrIco_2}   & 2 & \zeta_5^3+\zeta_5^2 & \zeta_5^4+\zeta_5 & -1 & 0 & 1 & -\zeta_5^3-\zeta_5^2 & -2 & -\zeta_5^4-\zeta_5 \\
\irrGAP_4 & \overline{\irrIco_3}  & 3 & -\zeta_5^3-\zeta_5^2 & -\zeta_5^4-\zeta_5 & 0 & -1 & 0 & -\zeta_5^3-\zeta_5^2 & 3 & -\zeta_5^4-\zeta_5 \\
 \irrGAP_5 & \irrIco_3  & 3 & -\zeta_5^4-\zeta_5 & -\zeta_5^3-\zeta_5^2 & 0 & -1 & 0 & -\zeta_5^4-\zeta_5 & 3 & -\zeta_5^3-\zeta_5^2 \\
 \irrGAP_6 & \irrIco_{4'} & 4 & -1 & -1 & 1 & 0 & 1 & -1 & 4 & -1 \\
 \irrGAP_7 & \irrIco_4 & 4 & -1 & -1 & 1 & 0 & -1 & 1 & -4 & 1 \\
 \irrGAP_8 & \irrIco_5 & 5 & 0 & 0 & -1 & 1 & -1 & 0 & 5 & 0 \\
 \irrGAP_9 & \irrIco_6 & 6 & 1 & 1 & 0 & 0 & 0 & -1 & -6 & -1 \\ \bottomrule
\end{array}
$$
    \caption{Character table for qubit group $\ico$ taken from GAP as \texttt{PerfectGroup(120)}. The 1st row is the size of each conjugacy class. The 1st column is the name given in GAP. The 2nd column is the name we gave in \cite{us1}. The 3rd column is the dimension of the irrep $|\irrGAP_i|$. Note that $\zeta_k = \exp(2\pi i/k)$ is a $k$-th root of unity.
    }
\end{table}

\twocolumngrid

\section{GAP code for $\ico$ in $\U(2)$}

\begin{center}
\begin{tabular}{c}
\begin{lstlisting}[linewidth=8cm]
gap> g:=PerfectGroup(120);; # 2I in U(2)
gap> ct:=CharacterTable(g);;
gap> f:=Irr(ct)[2];; # fundamental 2-dim irrep in U(2)
gap> Degree(f);
2

# Code Family: X.3-twisted 2-group
gap> lambda:=Irr(ct)[3];; 
gap> Degree(lambda);
2
gap> Norm(lambda*f);  
1
gap> Norm(lambda*f^2); 
2
gap> Norm(lambda*f^3); 
6
gap> for n in [1..21] do
>       s:=ScalarProduct(lambda,f^n);
>       if s>0 then Print("\t\tn=",n,"\t(m_lambda=",s,")\n"); fi;
> od;
                n=7     (m_lambda=1)
                n=9     (m_lambda=8)
                n=11    (m_lambda=44)
                n=13    (m_lambda=209)
                n=15    (m_lambda=924)
                n=17    (m_lambda=3928)
                n=19    (m_lambda=16321)
                n=21    (m_lambda=66880)
\end{lstlisting}
\end{tabular}
\end{center}

\newpage 
Let $\G = \ico$ in  $\U(2)$. One can show the following branching rules. 
\begin{table}[htp]
    \centering
    \begin{tabular}{l|l} \toprule
        $\irrE$ & $\irre$ \\ \midrule
         $\bm{1}$ & $\irrGAP_1$ \\
       $\bm{3}$ & $\irrGAP_5$ \\
        $\bm{5}$ & $\irrGAP_8$ \\
        $\bm{7}$ & $\irrGAP_4 \oplus \irrGAP_6$ \\
       $\bm{9}$ & $\irrGAP_6 \oplus \irrGAP_8$ \\
        $\bm{11}$ & $\irrGAP_4 \oplus \irrGAP_5 \oplus \irrGAP_8$ \\
        $\bm{13}$ & $\irrGAP_1 \oplus \irrGAP_5 \oplus \irrGAP_6 \oplus \irrGAP_8$ \\
       $\bm{15}$ & $\irrGAP_4 \oplus \irrGAP_5 \oplus \irrGAP_6 \oplus \irrGAP_8$ \\
       $\bm{17}$ & $\irrGAP_4 \oplus \irrGAP_6 \oplus (2) \irrGAP_8$ \\
        $\bm{19}$ & $\irrGAP_4 \oplus \irrGAP_5 \oplus (2) \irrGAP_6 \oplus \irrGAP_8$ \\
        $\bm{21}$ & $\irrGAP_1 \oplus \irrGAP_4 \oplus \irrGAP_5 \oplus \irrGAP_6 \oplus (2) \irrGAP_8$ \\ \bottomrule
    \end{tabular}
    \caption{Branching errors $\U(2) \downarrow \ico $}
\end{table}

Recall that for $\U(2)$ we had $\mathcal{E}_t = \{ \irrep{1}, \irrep{3}, \irrep{5}, \cdots, (\irrep{2t+1}) \}$. Also recall from \cref{lem:unitarytgroup1} that $\G$ is a unitary $t$-group iff $\expval{1, \irrechar} = 0$ for each $\irrE$ in $\mathcal{E}_t$ such that $\irrE\neq \irrep{1}$. From this we see that the smallest non-trivial $\irrE$ such that $\irre$ overlaps with $\irrGAP_1 = \irrep{1}$ is when $\irrE = \irrep{13}$ which is in $\mathcal{E}_6$. But there are no copies of $\irrGAP_1 = \irrep{1}$ (coming from non-trivial irreps) in either $\mathcal{E}_1$, $\mathcal{E}_2$, $\mathcal{E}_3$, $\mathcal{E}_4$, or $\mathcal{E}_5$. Thus we immediately see that $\ico$ is a unitary $5$-group. 

On the other hand, recall from \cref{lem:twistedunitarytgroup1} that $\G$ is a $\logirr$-twisted unitary $t$-group iff $\expval{\lambda \lambda^*, \irrechar} = 0$ for each $\irrE$ in $\mathcal{E}_t$ such that $\irrE\neq \irrep{1}$. If $\logirr = \irrGAP_3$ then $\logirr \otimes \logirr^* = \irrGAP_1 \oplus \irrGAP_4$. From the table above, we see that the smallest non-trivial $\irrE$ such that $\irre$ overlaps with either $\irrGAP_1$ or $\irrGAP_4$ is when $\irrE = \irrep{7}$ which is in $\mathcal{E}_3$. But there are no copies of $\irrGAP_1$ or $\irrGAP_4$ (coming from non-trivial irreps) in either $\mathcal{E}_1$ or $\mathcal{E}_2$. This proves that $\ico$ is a $\irrGAP_3$-twisted unitary $2$-group.

\newpage 
\onecolumngrid
\section{ {\large Example 2: $\Sigma(360\phi)$ Qutrit Codes}}

\begin{table}[htp]
    \centering
$$
\hspace*{-0.5cm}
\begin{array}{c|ccccccccccccccccc} \toprule 
& 1 & 1 & 1 & 45 & 45 & 45 & 120 & 120 & 90 & 90 & 90 & 72 & 72 & 72 & 72 & 72 & 72 \\ \midrule 
\irrGAP_1 & 1 & 1 & 1 & 1 & 1 & 1 & 1 & 1 & 1 & 1 & 1 & 1 & 1 & 1 & 1 & 1 & 1 \\ 
\irrGAP_2 & 3 & 3 \zeta_3^2 & 3 \zeta_3 & -1 & -\zeta_3^2 & -\zeta_3 & 0 & 0 & 1 & \zeta_3^2 & \zeta_3 & -\zeta_{15}^8-\zeta_{15}^2 & -\zeta_5^4-\zeta_5 & -\zeta_{15}^{13}-\zeta_{15}^7 & -\zeta
   _{15}^4-\zeta_{15} & -\zeta_{15}^{14}-\zeta_{15}^{11} & -\zeta_5^3-\zeta_5^2 \\
\irrGAP_3 & 3 & 3 \zeta_3^2 & 3 \zeta_3 & -1 & -\zeta_3^2 & -\zeta_3 & 0 & 0 & 1 & \zeta_3^2 & \zeta_3 & -\zeta_{15}^{14}-\zeta_{15}^{11} & -\zeta_5^3-\zeta_5^2 & -\zeta_{15}^4-\zeta_{15} &
   -\zeta_{15}^{13}-\zeta_{15}^7 & -\zeta_{15}^8-\zeta_{15}^2 & -\zeta_5^4-\zeta_5 \\
\irrGAP_4 & 3 & 3 \zeta_3 & 3 \zeta_3^2 & -1 & -\zeta_3 & -\zeta_3^2 & 0 & 0 & 1 & \zeta_3 & \zeta_3^2 & -\zeta_{15}^4-\zeta_{15} & -\zeta_5^3-\zeta_5^2 & -\zeta_{15}^{14}-\zeta_{15}^{11} &
   -\zeta_{15}^8-\zeta_{15}^2 & -\zeta_{15}^{13}-\zeta_{15}^7 & -\zeta_5^4-\zeta_5 \\
\irrGAP_5 & 3 & 3 \zeta_3 & 3 \zeta_3^2 & -1 & -\zeta_3 & -\zeta_3^2 & 0 & 0 & 1 & \zeta_3 & \zeta_3^2 & -\zeta_{15}^{13}-\zeta_{15}^7 & -\zeta_5^4-\zeta_5 & -\zeta_{15}^8-\zeta_{15}^2 & -\zeta
   _{15}^{14}-\zeta_{15}^{11} & -\zeta_{15}^4-\zeta_{15} & -\zeta_5^3-\zeta_5^2 \\
\irrGAP_6 & 5 & 5 & 5 & 1 & 1 & 1 & -1 & 2 & -1 & -1 & -1 & 0 & 0 & 0 & 0 & 0 & 0 \\
\irrGAP_7 & 5 & 5 & 5 & 1 & 1 & 1 & 2 & -1 & -1 & -1 & -1 & 0 & 0 & 0 & 0 & 0 & 0 \\
\irrGAP_8 & 6 & 6 \zeta_3^2 & 6 \zeta_3 & 2 & 2 \zeta_3^2 & 2 \zeta_3 & 0 & 0 & 0 & 0 & 0 & \zeta_3 & 1 & \zeta_3^2 & \zeta_3^2 & \zeta_3 & 1 \\
\irrGAP_9 & 6 & 6 \zeta_3 & 6 \zeta_3^2 & 2 & 2 \zeta_3 & 2 \zeta_3^2 & 0 & 0 & 0 & 0 & 0 & \zeta_3^2 & 1 & \zeta_3 & \zeta_3 & \zeta_3^2 & 1 \\
\irrGAP_{10} & 8 & 8 & 8 & 0 & 0 & 0 & -1 & -1 & 0 & 0 & 0 & -\zeta_5^4-\zeta_5 & -\zeta_5^4-\zeta_5 & -\zeta_5^4-\zeta_5 & -\zeta_5^3-\zeta_5^2 & -\zeta_5^3-\zeta_5^2 & -\zeta_5^3-\zeta_5^2 \\
\irrGAP_{11} & 8 & 8 & 8 & 0 & 0 & 0 & -1 & -1 & 0 & 0 & 0 & -\zeta_5^3-\zeta_5^2 & -\zeta_5^3-\zeta_5^2 & -\zeta_5^3-\zeta_5^2 & -\zeta_5^4-\zeta_5 & -\zeta_5^4-\zeta_5 & -\zeta_5^4-\zeta_5 \\
\irrGAP_{12} & 9 & 9 & 9 & 1 & 1 & 1 & 0 & 0 & 1 & 1 & 1 & -1 & -1 & -1 & -1 & -1 & -1 \\
\irrGAP_{13} & 9 & 9 \zeta_3^2 & 9 \zeta_3 & 1 & \zeta_3^2 & \zeta_3 & 0 & 0 & 1 & \zeta_3^2 & \zeta_3 & -\zeta_3 & -1 & -\zeta_3^2 & -\zeta_3^2 & -\zeta_3 & -1 \\
\irrGAP_{14} & 9 & 9 \zeta_3 & 9 \zeta_3^2 & 1 & \zeta_3 & \zeta_3^2 & 0 & 0 & 1 & \zeta_3 & \zeta_3^2 & -\zeta_3^2 & -1 & -\zeta_3 & -\zeta_3 & -\zeta_3^2 & -1 \\
\irrGAP_{15} & 10 & 10 & 10 & -2 & -2 & -2 & 1 & 1 & 0 & 0 & 0 & 0 & 0 & 0 & 0 & 0 & 0 \\
\irrGAP_{16} & 15 & 15 \zeta_3^2 & 15 \zeta_3 & -1 & -\zeta_3^2 & -\zeta_3 & 0 & 0 & -1 & -\zeta_3^2 & -\zeta_3 & 0 & 0 & 0 & 0 & 0 & 0 \\
\irrGAP_{17} & 15 & 15 \zeta_3 & 15 \zeta_3^2 & -1 & -\zeta_3 & -\zeta_3^2 & 0 & 0 & -1 & -\zeta_3 & -\zeta_3^2 & 0 & 0 & 0 & 0 & 0 & 0 \\ \bottomrule
\end{array}
$$
 \caption{Character table for qutrit group $\Sigma(360\phi)$ taken from GAP as \texttt{PerfectGroup(1080)}. The 1st row is the size of each conjugacy class. The 1st column is the name given in GAP. The 2nd column is the dimension of the irrep $|\irrGAP_i|$. Note that $\zeta_k = \exp(2\pi i/k)$ is a $k$-th root of unity.}
\end{table}

\twocolumngrid

Let $\G = \Sigma(360\phi)$ in $\U(3)$. The fundamental irrep is $\F = \irrep{3}$ and the anti-fundamental irrep is $\F^* = \irrepbar{3}$. Then using \cite{LieART} we can compute:
\begin{align*}
    \qty( \irrep{3} \otimes \irrepbar{3})^{\otimes 1} &= \irrep{1} \oplus \irrep{8} \numberthis \\
    \qty( \irrep{3} \otimes \irrepbar{3})^{\otimes 2} &= (2)\irrep{1} \oplus (4)\irrep{8} \oplus \irrep{10} \oplus \irrepbar{10} \oplus \irrep{27} \numberthis \\
     \qty( \irrep{3} \otimes \irrepbar{3})^{\otimes 3} &= (6)\irrep{1} \oplus (17)\irrep{8} \oplus (7)\irrep{10} \oplus (7)\irrepbar{10} \\
     & \qquad \oplus (9)\irrep{27}  \oplus (2)\irrep{35} \oplus (2)\irrepbar{35} \oplus \irrep{64} \numberthis \\
     \qty( \irrep{3} \otimes \irrepbar{3})^{\otimes 4} &= (23)\irrep{1} \oplus (80) \irrep{8} \oplus (42) \irrep{10} \oplus (42)\irrepbar{10} \\
     &\qquad \oplus (63) \irrep{27} \oplus (2) \irrep{28} \oplus (2) \irrepbar{28} \\
     &\qquad \oplus (23) \irrep{35} \oplus (23) \irrepbar{35} \oplus (16) \irrep{64}\\
     &\qquad \oplus (3) \irrep{81} \oplus (3) \irrepbar{81} \oplus \irrep{125}. \numberthis
\end{align*}
Thus we see that
\begin{align}
    \mathcal{E}_1 &= \{ \irrep{1}, \irrep{8} \} \\
    \mathcal{E}_2 &= \{ \irrep{1}, \irrep{8}, \irrep{10}, \irrepbar{10}, \irrep{27} \} \\
    \mathcal{E}_3 &= \{ \irrep{1}, \irrep{8}, \irrep{10}, \irrepbar{10}, \irrep{27}, \irrep{35}, \irrepbar{35}, \irrep{64} \} \\
    \mathcal{E}_4 &= \{ \irrep{1}, \irrep{8}, \irrep{10}, \irrepbar{10}, \irrep{27}, \irrep{28}, \irrepbar{28}, \irrep{35}, \irrepbar{35}, \irrep{64} , \irrep{81}, \irrepbar{81}, \irrep{125} \}.
\end{align}

Recall from \cref{lem:unitarytgroup1} that $\G$ is a unitary $t$-group iff $\expval{1, \irrechar} = 0$ for each $\irrE$ in $\mathcal{E}_t$ such that $\irrE\neq \irrep{1}$. From this we see that the smallest non-trivial $\irrE$ such that $\irre$ overlaps with $\irrGAP_1 = \irrep{1}$ is when $\irrE = \irrep{28}$ (or $\irrepbar{28}$ or $\irrep{125}$) which is in $\mathcal{E}_4$. But there are no copies of $\irrGAP_1 = \irrep{1}$ (coming from non-trivial irreps) in either $\mathcal{E}_1$, $\mathcal{E}_2$, or $\mathcal{E}_3$. Thus we immediately see that $\Sigma(360\phi)$ is a unitary $3$-group.

\begin{table}[htp]
    \centering
    $$
    \begin{array}{l|l} \toprule
        \irrE & \irre \\ \midrule
         \irrep{1} & \irrGAP_1 \\
         \irrep{8} & \irrGAP_{10} \\ \midrule
         \irrep{10}, \irrepbar{10} & \irrGAP_{15} \\
         \irrep{27} & \irrGAP_6 \oplus \irrGAP_7 \oplus \irrGAP_{11} \oplus \irrGAP_{12} \\ \midrule
         \irrep{35}, \irrepbar{35} & \irrGAP_{10} \oplus \irrGAP_{11} \oplus \irrGAP_{12} \oplus \irrGAP_{15} \\
         \irrep{64} & \irrGAP_6 \oplus \irrGAP_7 \oplus \irrGAP_{10} \oplus \irrGAP_{11} \oplus (2) \irrGAP_{12} \oplus (2) \irrGAP_{15}  \\ \midrule
         \irrep{28}, \irrepbar{28} & \irrGAP_1 \oplus \irrGAP_6 \oplus \irrGAP_7 \oplus \irrGAP_{10} \oplus \irrGAP_{12} \\
         \irrep{81}, \irrepbar{81} & \irrGAP_6 \oplus \irrGAP_7 \oplus (2) \irrGAP_{10} \oplus (2) \irrGAP_{11} \oplus \irrGAP_{12} \oplus (3) \irrGAP_{15} \\
         \irrep{125} & \irrGAP_1 \oplus (2) \irrGAP_6 \oplus (2) \irrGAP_7 \oplus (3) \irrGAP_{10} \oplus (3) \irrGAP_{11} \oplus (4) \irrGAP_{12} \oplus (2) \irrGAP_{15} \\ \bottomrule
    \end{array}
    $$
    \caption{Branching $\U(3) \downarrow \Sigma(360\phi) $}
\end{table}

On the other hand, recall from \cref{lem:twistedunitarytgroup1} that $\G$ is a $\logirr$-twisted unitary $t$-group iff $\expval{\lambda \lambda^*, \irrechar} = 0$ for each $\irrE$ in $\mathcal{E}_t$ such that $\irrE\neq \irrep{1}$. If $\logirr = \irrGAP_3$ or $ \logirr = \irrGAP_4$ then $\logirr \otimes \logirr^* = \irrGAP_1 \oplus \irrGAP_{11}$. From the table above, we see that the smallest non-trivial $\irrE$ such that $\irre$ overlaps with either $\irrGAP_1$ or $\irrGAP_{11}$ is when $\irrE = \irrep{27}$ which is in $\mathcal{E}_2$. But there are no copies of $\irrGAP_1$ or $\irrGAP_{11}$ (coming from non-trivial irreps) in $\mathcal{E}_1$. This proves that $\Sigma(360\phi)$ is both a $\irrGAP_3$-twisted unitary $1$-group and a $\irrGAP_4$-twisted unitary $1$-group.

\onecolumngrid
\section{GAP code for $\Sigma(360\phi)$ in $\U(3)$}
\begin{center}
\begin{tabular}{c}
\begin{lstlisting}[linewidth=12.4cm]
gap> g:=PerfectGroup(1080);; # Sigma(360phi) in U(3)
gap> ct:=CharacterTable(g);;
gap> f:=Irr(ct)[2];; # fundamental 3-dim irrep in U(3)
gap> Degree(f);
3

# Code family 1: X.3-twisted 1-group
gap> lambda1:=Irr(ct)[3];; # 3-dim code irrep
gap> Degree(lambda1);
3
gap> Norm(lambda1*f); 
1
gap> Norm(lambda1*f^2); 
3
gap> for n in [1..21] do
>       s:=ScalarProduct(lambda1,f^n);
>       if s>0 then Print("\t\tn=",n,"\t(m_lambda=",s,")\n"); fi;
> od;
                n=7     (m_lambda=15)
                n=10    (m_lambda=477)
                n=13    (m_lambda=13222)
                n=16    (m_lambda=358450)
                n=19    (m_lambda=9684357)


# Code family 2: X.4-twisted 1-group
gap> lambda2:=Irr(ct)[4];; # 3-dim code irrep
gap> Degree(lambda2);
3
gap> Norm(lambda2*f); 
1
gap> Norm(lambda2*f^2); 
3
gap> for n in [1..21] do
>       s:=ScalarProduct(lambda2,f^n);
>       if s>0 then Print("\t\tn=",n,"\t(m_lambda=",s,")\n"); fi;
> od;
                n=5     (m_lambda=1)
                n=8     (m_lambda=49)
                n=11    (m_lambda=1452)
                n=14    (m_lambda=39754)
                n=17    (m_lambda=1075727)
                n=20    (m_lambda=29054667)
\end{lstlisting}
\end{tabular}
\end{center}

\end{document}